\documentclass{amsart}

\usepackage{amsrefs,slashed}

\newtheorem{theorem}{Theorem}[section]
\newtheorem{lemma}[theorem]{Lemma}

\newtheorem{corollary}[theorem]{Corollary}

\theoremstyle{definition}
\newtheorem{definition}[theorem]{Definition}

\theoremstyle{remark}
\newtheorem{remark}[theorem]{Remark}
\newtheorem*{NB}{N.B}

\numberwithin{equation}{section}

\newcommand{\bC}{\mathbb{C}}
\newcommand{\bK}{\mathbb{K}}
\newcommand{\bN}{\mathbb{N}}
\newcommand{\bR}{\mathbb{R}}
\newcommand{\bZ}{\mathbb{Z}}
\newcommand{\cA}{\mathcal{A}}
\newcommand{\cB}{\mathcal{B}}

\newcommand{\cD}{\mathcal{D}}
\newcommand{\cH}{\mathcal{H}}
\newcommand{\cHi}{\mathcal{H}^{\infty}}
\newcommand{\cK}{\mathcal{K}}
\newcommand{\cL}{\mathcal{L}}
\newcommand{\cP}{\mathcal{P}}

\newcommand{\gtimes}{\hat{\otimes}}

\newcommand{\abs}[1]{\left|{#1}\right|}
\newcommand{\norm}[1]{\left\|{#1}\right\|}
\newcommand{\ip}[1]{\left\langle {#1}\right\rangle}
\newcommand{\hp}[1]{\left({#1}\right)}
\newcommand{\ncint}{{\int \!\!\!\!\!\! -}}
\newcommand{\set}[1]{\left\{{#1}\right\}}
\newcommand{\inv}[1]{{#1}^{-1}}

\newcommand{\Cstar}{{\ensuremath{C^*}}}

\newcommand{\latin}[1]{{\it{#1}\/}}

\newcommand{\Star}{\ensuremath{\ast}}
\newcommand{\term}[1]{{\it{#1}\/}}

\DeclareMathOperator{\dom}{Dom}
\DeclareMathOperator{\Cl}{Cl}
\DeclareMathOperator{\bCl}{\mathbb{C}l}
\DeclareMathOperator{\End}{End}
\DeclareMathOperator{\Id}{id}

\begin{document}

\title[Almost-commutative spectral triples]{A reconstruction theorem for almost-commutative spectral triples}

\author{Branimir \'Ca\'ci\'c}
\address{Department of Mathematics\\ California Institute of Technology\\ MC 253-37\\1200 E California Blvd\\ Pasadena, CA 91125}
\email{branimir@caltech.edu}



\date{June 28, 2012}


\begin{abstract}
We propose an expansion of the definition of almost-commutative spectral triple that accommodates non-trivial fibrations and is stable under inner fluctuation of the metric, and then prove a reconstruction theorem for almost-commutative spectral triples under this definition as a simple consequence of Connes's reconstruction theorem for commutative spectral triples. Along the way, we weaken the orientability hypothesis in the reconstruction theorem for commutative spectral triples, and following Chakraborty and Mathai, prove a number of results concerning the stability of properties of spectral triples under suitable perturbation of the Dirac operator.
\end{abstract}

\maketitle

\section{Introduction}

The commutative Gel'fand-Na{\u\i}mark theorem, which constructs a contravariant equivalence between the category of compact Hausdorff spaces and continuous maps, and the category of commutative unital \Cstar-algebras and \Star-homomorphisms, has motivated the identification in noncommutative geometry of the category of \Cstar-algebras as a category of noncommutative topological spaces. However, noncommutative geometers have not yet settled upon a category of noncommutative Riemannian manifolds.

Nevertheless, Connes has proposed a candidate for the objects of such a category, motivated by the example of the Atiyah--Singer Dirac operator of a compact spin manifold: \emph{spectral triples}~\cite{Con95}. On the one hand, every compact oriented Riemannian manifold can be seen to give rise to a spectral triple, for every such manifold $M$ gives rise to symmetric Dirac-type operators on self-adjoint Clifford module bundles over $M$. On the other hand, after substantial attempts by Rennie and V\'arilly~\cite{RV06}, Connes proved~\cite{Con08} in 2008 the so-called reconstruction theorem for commutative spectral triples, a result conjectured by him in 1996~\cite{Con96} that shows that a spectral triple with commutative algebra and satisfying certain conditions necessarily arises from a compact oriented Riemannian manifold.

Soon after introducing spectral triples, Connes also introduced \term{almost-commuta\-tive} spectral triples for the express purpose of reformulating the [classical field theory of the] Standard Model in noncommutative-geometric terms~\cites{Con95b,Con96}, a project that culminated in 2006 with the near-simultaneous papers by Barrett~\cite{Bar07} and by Chamsedinne, Connes and Marcolli~\cite{CCM07} (see also \cites{Con06,CoMa}). Such spectral triples are defined as the noncommutative-geometric Cartesian product of the canonical spectral triple of a compact spin manifold, the prototypical \term{commutative} spectral triple, with a \term{finite} spectral triple, namely, a spectral triple with finite-dimensional Hilbert space (see~\cites{Kra98,PS98,Ca10} for the general theory, and~\cites{ACG1,ACG2,ACG3,ACG4,ACG5} for classification results).

In 2008, Connes suggested that a reconstruction theorem should be feasible for almost-commutative spectral triples as well~\cite{Con08b}. In this paper, we prove just such a result, indeed, as a straightforward consequence of the reconstruction theorem for commutative spectral triples. However, this requires expanding the definition of almost-commutative spectral triples to include, roughly speaking, finite fibrations of commutative spectral triples, instead of only Cartesian products of commutative spectral triples with finite spectral triples. Indeed, we propose to define an \term{almost-commutative spectral triple} as a triple of the form $(C^{\infty}(X,A),L^{2}(X,H),D)$, where:
\begin{itemize}
	\item $H$ is a (possibly $\bZ/2\bZ$-graded) self-adjoint Clifford module bundle over a compact oriented Riemannian manifold $X$;
	\item $D$ is a symmetric Dirac-type operator on $H$;
	\item $A$ is a real unital \Star-algebra sub-bundle of the bundle $\End_{\Cl(X)}^{+}(H)$ of even endomorphisms of $H$ supercommuting with the Clifford action on $H$ defined by $D$. 
\end{itemize}  
This expanded definition still allows for all the familiar global-analytic tools for computing the spectral action on almost-commutative triples. Moreover, unlike the conventional definition, it is stable under inner fluctuations of the metric, and encompasses a number of global-analytically defined spectral triples already discussed in the literature, for instance, by Zhang~\cite{Zha10} and by Boeijink and van Suijlekom~\cite{BvS10}.

Along the way, we also show how to weaken the orientability hypothesis on the reconstruction theorem for commutative spectral triples, by only requiring that the action of the orientability Hochschild cycle anticommute with noncommutative $1$-forms in the even-dimensional case, and commute with noncommutative $1$-forms in the odd dimensional case. Moreover, following Chakraborty and Mathai~\cite{ChM09}, we provide proofs of a number of folkloric results on the stability of certain properties of spectral triples under suitable perturbation of the Dirac operator.

The author would like to thank his advisor, Matilde Marcolli, for her extensive comments and for her advice, support and patience, Partha Sarathi Chakraborty and Nigel Higson for technical advice on some of the stability results, and Nikola\u \i\ Ivankov, Helge Kr\"uger, Bram Mesland, Kevin Teh, Rafael Torres and Dapeng Zhang for helpful comments and conversations. The author also gratefully acknowledges the financial and administrative support of the Department of Mathematics of the California Institute of Technology, as well as the hospitality and support of the Max Planck Institute for Mathematics.

\section{Definitions and results}

In the following, we shall motivate and propose both concrete (i.e., global-analytic) and abstract (i.e., noncommutative-geometric) definitions of almost-com\-mutative spectral triples, and then state and prove a reconstruction theorem for almost-commutative spectral triples, thereby establishing the equivalence of the concrete and abstract definitions.

\subsection{Dirac-type operators} Let us first recall a few definitions and facts from the theory of Dirac-type operators, mostly to establish notation and terminology; for a full account, see~\citelist{\cite{BGV}*{\S\S~3.1--3} \cite{GBVF}*{Chapter 5, \S\S~9.1--3}}. Throughout, we shall use the conventions of super-linear algebra~\cite{BGV}*{\S~1.3}. Thus, if $H$ is a $\bZ/2\bZ$-graded vector bundle, we consider $\End(H)$ as $\bZ/2\bZ$-graded as well, and we shall consider the $\bZ/2\bZ$-graded tensor product of $\bZ/2\bZ$-graded vector bundles and operators on them, which we denote by $\gtimes$. If a vector bundle is not explicitly $\bZ/2\bZ$-graded, as shall often be the case, we consider it as trivially $\bZ/2\bZ$-graded.

Now recall that, given a Riemannian manifold $(X,g)$, one can define the Clifford bundle $\Cl(X)$ of $(X,g)$ and thus consider Clifford module bundles over $(X,g)$; in particular, a Clifford module bundle is \term{self-adjoint} if it is Hermitian as a vector bundle, and if $1$-forms act as skew-symmetric bundle endomorphisms. If $H$ is a self-adjoint Clifford module bundle, we can define the unital \Star-algebra sub-bundle $\End_{\Cl(X)}^{+}(H)$ of $\End(H)$ whose sections are those even bundle endomorphisms that supercommute (and thus commute) with the Clifford action on $H$.

At last, we can recall the following definition, which will be central to our discussion of almost-commutative spectral triples:

\begin{definition}
Let $E$ be a Clifford module bundle over a Riemannian manifold $(X,g)$. A \term{Dirac-type operator} on $E$ is a first-order differential operator on $E$ such that \[[D,f] = c(df), \quad f \in C^{\infty}(M),\] where $c :  \Cl(M) \to \End(E)$ denotes the Clifford action on $E$.
\end{definition}

For more on Dirac-type operators, see the monographs by Roe~\cite{Roe} and by Berline, Getzler and Vergne~\cite{BGV}, as well as the notes by Roepstorff and Vehns~\cites{RV1,RV2}. Note that~\cite{BGV} and~\cites{RV1,RV2} consider only odd Dirac-type operators on $\bZ/2\bZ$-graded Clifford module bundles, a restriction that is unnecessary for our purposes.

\subsection{Almost-commutative spectral triples}\label{cast}

In order to motivate our new definitions, let us recall the conventional definition of almost-commutative spectral triple, or, for convenience, \term{product almost-commutative spectral triples}.

Let $X$ be a compact spin manifold with fixed spinor bundle $S$ and corresponding Atiyah--Singer Dirac operator $\slashed{D}$, and let $F :=(\cA_{F},\cH_{F},D_{F})$ be a finite spectral triple, that is, a spectral triple with finite-dimensional Hilbert space $\cH_{F}$. If $X$ is even-dimensional, view $L^{2}(X,S)$ as $\bZ/2\bZ$-graded by the chirality operator, so that $C^{\infty}(X)$ acts on $L^{2}(X,S)$ by even operators, and $\slashed{D}$ is an odd operator. One can therefore define the noncommutative-geometric \term{Cartesian product} of $X$ and $F$ by
\[
	X \times F := (C^{\infty}(X) \gtimes \cA_{F}, L^{2}(X,S)\gtimes \cH_{F}, \slashed{D} \gtimes 1 + 1 \gtimes D_{F});
\]
a product almost-commutative spectral triple, then, is a spectral triple of this form. 

\begin{remark}
Note that this definition is not stable under inner fluctuations of the metric, for if $M = \sum_{i=1}^{n} a_{i}[D,b_{i}]$ for non-constant $a_{i}$, $b_{i} \in C^{\infty}(X) \otimes \cA_{F}$, then $M$ is generally not of the form $1 \gtimes T$ for some constant $T \in B(\cH_{F})$.
\end{remark}

Now, viewing $X \times \cH_{F}$ as a globally trivial $\bZ/2\bZ$-graded Hermitian vector bundle, and $X \times \cA_{F}$ as globally trivial unital \Star-algebra bundle, indeed, as a real unital \Star-algebra sub-bundle of $\End^{+}(X \times \cH_{F})$, define a self-adjoint Clifford module bundle $H$ and a \Star-algebra subbundle $A$ of $\End^{+}_{\Cl(X)}(H)$ by
\[
	H := S \gtimes (X \times \cH_{F}), \quad A := L \gtimes (X \times\cA_{F}),
\]
where $L$ is the trivial real unital \Star-algebra sub-bundle of $\End^{+}(S)$ defined by $L_{x} := \bR 1_{S_{x}}$. Finally, let $D = \slashed{D}_{X \times \cH_{F}} + 1 \gtimes D_{F}$, where $\slashed{D}_{X \times \cH_{F}}$ is the twisted Dirac operator on the twisted spinor bundle $H$ corresponding to the trivial connection on $X \times \cH_{F}$. Then $D$ is a symmetric Dirac-type operator on the self-adjoint Clifford module bundle $H$, and $A$ is a \Star-algebra subbundle of $\End_{\Cl(X)}^{+}(H)$, such that
\[
	X \times F = (C^{\infty}(X,A), L^{2}(X,H), D).
\]

On the other hand, suppose that $X$ is a compact oriented Riemannian manifold, $H$ is a self-adjoint Clifford module bundle over $X$, $A$ is a unital \Star-algebra sub-bundle of $\End_{\Cl(X)}^{+}(H)$, and $D$ is a symmetric Dirac-type operator on $H$. Then by standard analytic results about Dirac-type operators~\cite{Hig06}*{Theorem 3.23}, together with the fact that sections of $A$ are even bundle endomorphisms supercommuting with the Clifford action $H$, so that $[D,a]$ is a bundle endomorphism for all $a \in C^{\infty}(X,A)$, $(C^{\infty}(M,A),L^{2}(M,H),D)$ is a spectral triple of metric dimension $\dim X$. Thus, we may sensibly generalise the conventional definition of almost-commutative spectral triple as follows:

\begin{definition}
An \term{almost-commutative spectral triple} is a spectral triple of the form $(C^{\infty}(X,A), L^{2}(X,H), D)$, where $X$ is a compact oriented Riemannian manifold, $H$ is a self-adjoint Clifford module bundle, $A$ is a real unital \Star-algebra sub-bundle of $\End_{\Cl(X)}^{+}(H)$, and $D$ is a symmetric Dirac-type operator on $H$.
\end{definition}

\begin{remark}
This definition is stable under inner fluctuation of the metric, for a perturbation of a symmetric Dirac-type operator by a symmetric bundle endomorphism is a symmetric Dirac-type operator inducing the same Clifford action.
\end{remark}

Since the square of a Dirac-type operator is a generalised Laplacian, this definition manifestly lends itself to the perturbative computation of the spectral action~\cite{CC97} via heat kernel methods~\cite{Gil}  (see~\cite{vdD11} for a comprehensive account for product almost-commutative spectral triples). Another feature of this definition is that it encompasses non-trivial ``fibrations'' in the following sense:

\begin{lemma}
 Let $X$ and $F$ be as above. Let $G$ be a compact Lie group, and let $\rho$ be an action of $G$ on $F$, namely, a unitary representation of $G$ on $\cH_{F}$ such that for each $g \in G$, $\rho(g)\cA_{F}\rho(g)^{\ast} \subset \cA_{F}$, and $\rho(g)D_{F}\rho(g)^{\ast}=D_{F}$; if $F$ is even, we moreover require the action of $G$ to commute with the $\bZ/2\bZ$-grading. Let $\cP$ be a principal $G$-bundle over $X$, and let $\nabla_{\cP}$ be a connection on $\cP$. Define $H$ and $A$ by
\[
	H := S \gtimes (\cP \times_{\rho} \cH_{F}), \quad A := L \gtimes (\cP \times_{\rho} \cA_{F}),
\]
where $L$ is defined as above, and let $D = \slashed{D}_{\cP \times_{\rho} \cH_{F}} + 1 \gtimes D_{F}$, where $\slashed{D}_{\cP \times_{\rho} \cH_{F}}$ is the twisted Dirac operator on the twisted spinor bundle $H$ corresponding to the connection on $\cP \times_{\rho} \cH_{F}$ induced by $\nabla_{\cP}$. Then
\[
	 X \times_{(\cP,\nabla_{\cP})} F := (C^{\infty}(X,A), L^{2}(X,H), D)
\]
is an almost-commutative spectral triple.
\end{lemma}

\begin{proof}
It follows immediately that $(C^{\infty}(X), L^{2}(X,H),D)$ is at least an almost-com\-mu\-ta\-tive spectral triple. However, since
\[
	C^{\infty}(X,A) = 1 \gtimes C^{\infty}(X,\cP \times_{\rho} \cA_{F}),
\]
sections of $A$ are even bundle endomorphisms supercommuting with the Clifford action on $H$, so that $(C^{\infty}(X,A), L^{2}(X,H), D)$ is also an almost-commutative spectral triple.
\end{proof}

We can view $X \times_{(\cP,\nabla_{\cP})} F$ as the product of $X$ and $F$ \term{twisted} by $(\cP,\nabla_{\cP})$; a concrete example of this construction has already been studied in detail by Boeijink and van Suijlekom~\cite{BvS10} in connection with Yang-Mills theory. It is also worth noting that the data $(\cP \times_{\rho} \cH_{F}, \nabla_{\cP}, D_{F})$ can be viewed as defining a non-trivial morphism $X \times_{(\cP,\nabla_{\cP})} F \to X$ in the category of spectral triples proposed by Mesland~\cites{Me09,Me09b}.

\subsection{Abstract almost-commutative spectral triples}\label{acstsec}

Now, we shall give an abstract definition of almost-commutative spectral triple, which shall depend upon an abstract definition of commutative spectral triple, identical to that proposed by Connes~\cites{Con96,Con08}, except for a weakening of the orientability condition.

Before continuing, it is worth recalling the following:

\begin{definition}
A spectral triple $(\cA,\cH,D)$ is called \term{regular} if $\cA + [D,\cA] \subset \cap_{k} \dom \delta^{k}$ for for the derivation $\delta : T \mapsto [\abs{D},T]$ on $B(H)$, and is called \term{strongly regular} if, in addition,  $\End_{\cA}(\cap_{k} \dom D^{k}) \subset \cap_{k} \dom \delta^{k}$.
\end{definition}

Now, let $(C^{\infty}(X,A),L^{2}(X,H),D)$ be an almost-commutative spectral triple. We may just as well consider it as being composed of two pieces:
\begin{enumerate}
	\item An almost-commutative spectral triple $(C^{\infty}(X),L^{2}(X,H),D)$ with \emph{commutative} algebra.
	\item A real unital \Star-algebra sub-bundle of $\End_{\Cl(X)}^{+}(H)$.
\end{enumerate}
In order to obtain an abstract definition of almost-commutative spectral triple, it therefore suffices to translate these two components into the language of noncommutative geometry.

Let us first consider the first component, namely, an almost-commutative spectral triple of the form $(C^{\infty}(X),L^{2}(X,H),D)$, where $H$ is a self-adjoint Clifford module bundle over a compact orientable Riemannian manifold $X$, and $D$ is a symmetric Dirac-type operator on $H$. We have already seen that $(C^{\infty}(X), L^{2}(X,H), D)$ is a spectral triple of metric dimension $\dim X$; that it is in fact strongly regular follows from~\cite{Con08}*{proof of Theorem 11.4}. Still more is true---one can check that it satisfies the following definition, very slightly modified, as mentioned above, from a definition of Connes's~\cites{Con96,Con08}:

\begin{definition}
Let $(\cA,\cH,D)$ be a strongly regular spectral triple of metric dimension $p \in \bN$, such that $\cA$ is commutative. We call $(\cA,\cH,D)$ a \term{commutative spectral triple} if the following conditions hold:
\begin{enumerate}
	\item {\bf Order one}: For any $a$, $b \in \cA$, $[[D,a],b] = 0$.
	\item {\bf Pre-orientability}: There exists an antisymmetric Hochschild p-cycle $c \in Z_{p}(\cA,\cA)$ such that $\chi = \pi_D(c)$ is a self-adjoint unitary on $\cH$ satisfying $a\chi = \chi a$ and $[D,a]\chi = (-1)^{p+1} \chi[D,a]$ for all $a \in \cA$.
	\item {\bf Finiteness}: One has that $\cHi := \cap_{m} \dom D^{m}$ is finitely generated and projective as a $\cA$-module.
	\item {\bf Absolute continuity}: The $\cA$-module $\cHi$ admits a Hermitian structure $\hp{\cdot,\cdot}$ defined by the equality $\ip{\xi,a\eta} = \ncint a\hp{\xi,\eta}\abs{D}^{-p}$ for $a \in \cA$, $\xi$, $\eta \in \cHi$.
\end{enumerate}
Moreover, if $p$ is even and $\{D,\chi\} = 0$, or if $p$ is odd and $\chi = 1$, then we shall call $(\cA,\cH,D)$ \term{orientable}.
\end{definition}

\begin{remark}
When dealing with a commutative spectral triple $(\cA,\cH,D)$, we may assume without loss of generality that $\cA$ is a complex \Star-subalgebra of $B(\cH)$, by replacing $\cA$ with $\cA + i\cA \subset B(\cH)$ in the case that $\cA$ is only a real \Star-subalgebra of $B(\cH)$.
\end{remark}

Indeed, in the case of $(C^{\infty}(X),L^{2}(X,H),D)$:
\begin{enumerate}
	\item The order one condition follows precisely since $D$ is a first-order differential operator.
	\item Pre-orientability follows from~\cite{Con08}*{proof of Theorem 11.4}, where the Hoch\-schild cycle $c$ is constructed from the volume form on $X$, and thus acts as the chirality operator.
	\item Finiteness follows since $\cap_{m} \dom D^{m} = C^{\infty}(X,H)$, which in turn follows from Sobolev theory applied to the elliptic operator $D$.
	\item Absolute continuity follows from the Connes trace formula~\cite{Con88} applied to pseudodifferential operators on the Hermitian vector bundle $H$ of the form $a \abs{D}^{-n}$, where $a \in C^{\infty}(X)$ (see~\citelist{\cite{GBVF}*{\S~9.4} \cite{Roe}*{Chapter 8}} for details).
\end{enumerate}
Although pre-orientability is an immediate consequence of orientability, it is a strictly weaker condition. Indeed, our spectral triple$(C^{\infty}(X),L^{2}(X,H),D)$ need not be orientable, as Connes assumed and as indeed holds for the canonical spectral triple $(C^{\infty}(N),L^{2}(N,S), \slashed{D})$ of a compact spin manifold $N$ with spinor bundle $S$. In the case of such an even-dimensional compact spin manifold $N$, if $E$ is a non-trivially $\bZ/2\bZ$-graded Hermitian vector bundle over $N$, and if $\slashed{D}_{E}$ is the twisted Dirac operator on $S \gtimes E$ corresponding to any self-adjoint superconnection on $E$, then $(C^{\infty}(N), L^{2}(N,S \gtimes E), \slashed{D}_{E})$ is not orientable, for any Hochschild p-cycle will act on $S \gtimes E$ by a bundle endomorphism of the form
\[
	T \gtimes 1 \in C^{\infty}(N,\End(S)) \gtimes C^{\infty}(N,\End(E)) \cong C^{\infty}(N,\End(S \gtimes E))
\]
with $T$ even, so that it cannot distinguish between $S \gtimes E^{+}$ and $S \gtimes E^{-}$, and thus cannot act as the $\bZ/2\bZ$-grading on $S \gtimes E$. In the case of an odd-dimensional compact oriented Riemannian manifold $N$, one can readily construct self-adjoint Clifford module bundles $H \to N$ such that the chirality operator (and thus any other non-trivial action of a Hochschild $p$-cycle) defines a non-trivial $\bZ/2\bZ$-grading on $H$~\citelist{\cite{LM}*{\S II.5} \cite{PR}*{\S 8}}.
Now, our spectral triple $(C^{\infty}(X),L^{2}(X,H),D)$ has one further property, precisely because $D$ is a Dirac-type operator on the self-adjoint Clifford module bundle $H$. Since the Clifford action on $H$ is given by
\[
	df \mapsto [D,f], \quad f \in C^{\infty}(X),
\]
it follows, in particular, that
\[
	\forall f \in C^{\infty}(X), \quad [D,f]^{2} = -g(df,df) \in C^{\infty}(X).
\]
Hence, $(C^{\infty}(X),L^{2}(X,H),D)$ satisfies the following:

\begin{definition}
Let $(\cA,\cH,D)$ be a commutative spectral triple of metric dimension $p$. If
\[
	\forall a \in \cA, \quad [D,a]^{2} \in \cA,
\]
then we shall say that $(\cA,\cH,D)$ is of \term{Dirac type}.
\end{definition}

\begin{remark}
The above definition, which makes sense even in the case of a noncommutative real spectral triple, is already fairly suggestive of the global-analytic definition of Dirac-type operator, and moreover lends itself to the following restatement. Let $(\cA,\cH,D)$ be a commutative spectral triple, let $\cHi = \cap_{k} \dom D^{k}$, and let $H_{1}(\cA,\cA)$ be the first Hochschild homology group of $\cA$ as a $\cA$-bimodule. By the order one condition and strong regularity, one can therefore define, as it were, a generalised symbol map $\sigma_{D} : H_{1}(\cA,\cA) \to \End_{\cA}(\cH^{\infty}) \subset B(\cH)$ by
\[
	\sigma_{D}([a_{0} \otimes a_{1}]) := a_{0}[D,a_{1}].
\]
Then $(\cA,\cH,D)$ is of Dirac type if and only if
\[
	\forall \eta \in H_{1}(\cA,\cA), \quad \sigma_{D}(\eta)^{2} \in \cA,
\]
a condition that is rather suggestive not only of the definition of Dirac-type operator, but also of the closely related $K$-theoretic notion of Clifford symbol.
\end{remark}

\begin{remark}
By the proof of Corollary~\ref{weakrecon}, together with~\cite{GBVF}*{Lemma 11.6}, it also follows that we could have equivalently defined that a commutative spectral triple $(\cA,\cH,D)$ is of Dirac type whenever the \Star-subalgebra of $B(\cH)$ generated by
\[
	\begin{cases}
		\cA + [D,\cA] &\text{if $(\cA,\cH,D)$ has even metric dimension,}\\
		\cA + [D,\cA][D,\cA] &\text{if $(\cA,\cH,D)$ has odd metric dimension,}
	\end{cases}
\]
has centre $\cA$. This alternative definition can be viewed as a direct translation of the fact that if $X$ is a compact oriented Riemannian manifold, then $\bCl(X)$, the complexification of $\Cl(X)$ if $\dim X$ is even, and the complexification of $\Cl^{+}(X)$ if $\dim X$ is odd, is a complex unital \Star-algebra bundle with fibre $M_{\dim X}(\bC)$.
\end{remark}

\begin{remark}
Ignoring real structures, Gracia-Bond{\'\i}a--V{\'a}rilly--Figueroa define~\cite{GBVF}*{Definition 11.2} an \term{irreducible} commutative spectral triple as a commutative spectral triple $(\cA,\cH,D)$ such that no non-trivial projection in $B(\cH)$ commutes with $\cA$, $D$, and $\pi_{D}(c)$. The discussion of~\cite{GBVF}*{\S~11.3} implies that an irreducible commutative spectral triple is, in particular, of Dirac type. However, being of Dirac type is a strictly weaker condition, for if $H \to X$ is a self-adjoint Clifford module bundle over $X$ compact oriented Riemannian but disconnected, and if $D$ is a symmetric Dirac-type operator on $H$, then $(C^{\infty}(X),L^{2}(X,H),D)$ is of Dirac type, but not irreducible.
\end{remark}

Let us now turn to the second component of an almost-commutative spectral triple, the real unital \Star-algebra bundle $A$. In order to translate this datum into noncommutative-geometric terms, we will need some sort of Serre--Swan-type result for such bundles, or more precisely, for bundles of finite-dimensional \Cstar-algebras. It turns out that such a result does indeed exist, but requires a slightly weaker notion of algebra bundle than the one we would like to use:

\begin{definition}[cf. \cite{BvS10}*{Definition 3.1}]
Let $\bK = \bR$ or $\bC$. An \emph{weak algebra bundle} is an $\bK$-vector bundle $A \to X$ together with morphism of vector bundles $\mu : A \otimes A \to A$ covering $\Id_{X}$ such that
\[
	\mu \circ \left(\Id_{A} \otimes \mu\right) = \mu \circ \left(\mu \otimes \Id_{A}\right),
\]
thereby inducing an $\bK$-algebra structure on each fibre; $(A,\mu)$ is called \emph{unital} if there exists $1_{A} \in C^{\infty}(X,A)$ such that for all $\xi \in C^{\infty}(X,A)$,
\[
	\mu(1_{A} \otimes \xi) = \mu(\xi \otimes 1_{A}) = \xi.
\]
If, moreover, there exists a conjugate-linear vector bundle endomorphism $J$ of $A \to X$ such that $J^{2} = 1$, inducing the structure of a \Star-algebra on each fibre, then $A \to X$ is called \emph{involutive}. 
\end{definition}

One can define categories of [unital] [involutive] weak algebra bundles, by defining a morphism of weak algebra bundles to be vector bundle morphism $T : (A \to X) \to (B \to X)$ covering $\Id_{X}$ such that
\[
	T \circ \mu_{A} = \mu_{B} \circ \left(T \otimes T\right),
\]
thereby inducing algebra homomorphisms $T_{x} : A_{x} \to B_{x}$ on the fibres; unital and involutive morphisms are defined in the analogous way. Finally, one can define a \emph{weak algebra sub-bundle} of a weak algebra bundle $(A,\mu_{A}) \to X$ to be a sub-bundle $B \to X$ of $A$ such that $\mu_{A}(B \otimes B) \subset B$, so that $\mu_{A}$ restricts to a weak algebra bundle structure on $B$; again, unital and involutive weak algebra sub-bundles are defined in the analogous way.

\begin{remark}
Every [unital] [\Star-]algebra bundle is a [unital] [involutive] weak algebra bundle, but not every [unital] [involutive] weak algebra bundle is a[n] [unital] [\Star-]algebra bundle, that is, a locally trivial bundle of finite-dimensional unital [\Cstar-]algebras, for different fibres need not be isomorphic as [unital] [\Star-]algebras. However, a weak [unital] [involutive] algebra sub-bundle of a[n] [unital] [\Star-]algebra bundle is necessarily a[n] [unital] [\Star-]algebra sub-bundle; since we will always be dealing with unital involutive weak algebra subbundles of the endomorphism bundles of Hermitian vector bundles, the difference between weak algebra bundle and algebra bundle will not affect us anywhere.
\end{remark}

Now that we have our weakened notion of algebra bundle, we can state the Serre--Swan theorem for weak algebra bundles, recently proved by Boeijink and van Suijlekom:

\begin{theorem}[Serre--Swan for weak algebra bundles~\cite{BvS10}*{Theorem 3.8}]\label{SerreSwan}
The map
\[
	A \to C^{\infty}(X,A)
\]
defines an equivalence of categories between the category of [unital] [involutive] weak $\bK$-algebra bundles and the category of [unital] finitely-generated $C^{\infty}(X,\bK)$-module [\Star-]algebras over $\bK$.
\end{theorem}

We can now proceed to characterise our real unital \Star-algebra bundle $A$, or rather, the finitely-generated projective $C^{\infty}(X,\bR)$-module \Star-algebra $C^{\infty}(X,A)$, as follows. First, by Theorem~\ref{SerreSwan}, $A \to X$ is a unital \Star-algebra sub-bundle of $\End(H)$ if and only if $C^{\infty}(X,A)$ is a finitely-generated projective sub-$C^{\infty}(X)$-module \Star-algebra of $C^{\infty}(X,\End(E)) = \End_{C^{\infty}(X)}(C^{\infty}(X,H))$; in the even ($\bZ/2\bZ$-graded) case, we need simply specify in addition that $C^{\infty}(X,A)$ consists of even operators. Given this, $A \to X$ is a unital \Star-algebra sub-bundle of $\End_{\Cl(X)}^{+}(H)$ if and only if every section of $A$ supercommutes with the Clifford action, if and only if for all $a \in C^{\infty}(X,A)$, $b \in C^{\infty}(X)$,
\[
	0 = [c(db),a] = [[D,b],a].
\]
Thus, we may characterise $C^{\infty}(X,A)$ as an even finitely-generated projective sub-$C^{\infty}(X,\bR)$-module \Star-algebra of $\End_{C^{\infty}(X)}(C^{\infty}(X,H))$, satisfying the following generalised order-one condition:
\[
	\forall a \in C^{\infty}(X,A) \; b \in C^{\infty}(X), \quad  [[D,b],a] = 0. 
\]

Putting everything together, we therefore see that an almost-commuta\-tive spectral triple $(C^{\infty}(X,A),L^{2}(X,H),D)$ satisfies the following abstract definition:

\begin{definition}
Let $(\cA,\cH,D)$ be a spectral triple, and let $\cB$ be a central unital $\ast$-subalgebra of $\cA$. We call $(\cA,\cH,D)$ an \term{abstract almost-commutative spectral triple} with \term{base} $\cB$ if the following three conditions hold:
\begin{enumerate}
	\item $(\cB,\cH,D)$ is a commutative spectral triple of Dirac type.
	\item $\cA$ is an even finitely generated projective $\cB$-module $\ast$-subalgebra of the algebra $\End_{\cB+i\cB}(\cHi)$, for $\cHi = \cap_{k \in \bN} \dom D^{k}$.
	\item For all $a \in \cA$, $b \in \cB$, $[[D,b],a] = 0$.
\end{enumerate}
\end{definition}

\subsection{Reconstruction theorems}

We can now finally state our reconstruction theorem for almost-commutative spectral triples, which establishes the equivalence of our two definitions:

\begin{theorem}\label{acst}
Let $(\cA,\cH,D)$ be a spectral triple, and let $\cB$ be a central unital $\ast$-subalgebra of $\cA$. Then $(\cA,\cH,D)$ is an abstract almost-commutative spectral triple with base $\cB$ if and only if it is an almost-commutative spectral triple, that is, if and only if there exist a compact oriented Riemannian manifold $X$, a self-adjoint Clifford module bundle $H$ over $X$, and a real unital $\ast$-algebra sub-bundle $A$ of $\End_{\Cl(X)}^{+}(H)$, such that $\cB = C^{\infty}(X)$, $\cA = C^{\infty}(X,A)$, $\cH = L^{2}(X,H)$, and $D$ is an essentially self-adjoint Dirac-type operator on $H$.
\end{theorem}

Since every abstract almost-commutative spectral triple is built on a commutative spectral triple, one would expect that our reconstruction theorem is a consequence of the following result, conjectured by Connes in 1996~\cite{Con96} and finally proved by him in 2008~\cite{Con08}:

\begin{theorem}[Reconstruction Theorem~\cite{Con08}*{Theorem 11.3}]\label{recon}
Let $(\cA,\cH,D)$ be an orientable commutative spectral triple of metric dimension $p$. Then there exists a smooth compact oriented $p$-manifold $X$ such that $\cA = C^{\infty}(X)$.
\end{theorem}

This is indeed the case, but we must first find a way to weaken the orientability hypothesis to a pre-orientability hypothesis; we shall also find it useful to extract more information out of the reconstruction theorem for commutative spectral triples. Indeed, we have the following, which also incorporates the results of~\cite{GBVF}*{Chapter 11}:

\begin{corollary}\label{weakrecon}
Let $(\cA,\cH,D)$ be a commutative spectral triple of metric dimension $p$. Then there exists a smooth compact oriented $p$-manifold $X$ and a Hermitian vector bundle $H$ over $X$ such that $\cA = C^{\infty}(X)$, $\cH = L^{2}(X,H)$, and $D$ is an essentially self-adjoint first-order differential operator on $H$. Moreover, if $(\cA,\cH,D)$ is of Dirac type, then $H$ is a self-adjoint Clifford module bundle and $D$ is a Dirac-type operator.
\end{corollary}

\begin{NB}
In the following proof, we will use several technical results on the stability of certain properties of spectral triples under suitable perturbation of the Dirac operator found in the appendix.
\end{NB}

\begin{proof}
First, suppose that $p$ is even. Now, on $\cHi = \cap_{k} \dom D^{k}$, we can write $D = D_{0} + M$ for $D_{0} = \tfrac{1}{2}[D,\chi]\chi$ and $M = \tfrac{1}{2}\{D,\chi\}\chi$. By direct computation and strong regularity, one finds that $M \in \End_{\cA}(\cHi) \subset B(\cH)$ and is self-adjoint. Hence, $(\cA,\cH,D_{0})$ is a spectral triple by Lemma~\ref{lem1}, and is of metric dimension $p$ by Lemma~\ref{lem2}. Since $D_{0}^{2} - D^{2} = -MD_{0}-D_{0}M = -\tfrac{1}{2}(D^{2}-\chi D^{2} \chi)$ on $\cHi$, $[D_{0}^{2}-D^{2},T] \in \cD_{k+1}$ for any $T \in \cD_{k}$, where $\cD$ is the extended algebra of differential operators for $(\cA,\cH,D)$ (see Definition~\ref{algdif}). Hence, by Lemma~\ref{lem4}, $(\cA,\cH,D_{0})$ is strongly regular. We can now check the other axioms for a commutative spectral triple in turn:
\begin{enumerate}
	\item Since $M$ commutes with $\cA$, the order one condition continues to hold.
	\item Again, since $M$ commutes with $\cA$, pre-orientability continues to hold with the same $c$ and $\chi$, and $\{D_{0},\chi\} = 0$ by construction of $D_{0}$, yielding orientability.
	\item Since $(\cA,\cH,D)$ is strongly regular and since $M \in \End_{\cA}(\cHi)$, Corollary~\ref{lem4a} allows us to apply Lemma~\ref{lem3} and conclude that \[\cHi := \cap_{k} \dom D^{k} = \cap_{k} \dom D_{0}^{k},\] so that finiteness continues to hold.
	\item Since $M \in \End_{\cA}(\cHi)$, it follows that absolute continuity continues to hold by Lemma~\ref{lem5}, with $\ncint T \abs{D}^{-p} = \ncint T \abs{D_{0}}^{-p}$ for all $T \in B(\cH)$.
\end{enumerate}
Thus, $(\cA,\cH,D_{0})$ is indeed an orientable commutative spectral triple, so that we may apply the Reconstruction Theorem to $(\cA,\cH,D_{0})$ to obtain $X$.

Now, suppose that $p$ is odd. Write $D = D_0 + M$ for
\[
	D_0 := \frac{1}{2}(D + \chi D \chi), \quad M := \frac{1}{2}(D-\chi D \chi).
\]
Since $\chi$ commutes with elements of $[D,\cA]$, $M$ commutes with $\cA$, and hence $M$ is a self-adjoint element of $\End_{\cA}(\cHi)$. The argument for the even case, \emph{mutatis mutandis}, then shows that $(\cA,\cH,D_0)$ is still a commutative spectral triple of metric dimension $p$, with $\pi_{D_0}(c) = \pi_D(c) = \chi$ and $D_0 \chi = \chi D_0$. In particular, since $\chi$ commutes with $D_0$ and with all elements of $\cA$, $D_1 = \chi D_0$ is a self-adjoint operator on $\cH$ satisfying $D_1^2 = D_0^2$ and $[D_1,a] = \chi [D_0,a]$ for all $a \in \cA$. Because of this, all the axioms for a commutative spectral triple of metric dimension $p$ immediately follow for $(\cA,\cH,D_1)$ except for pre-orientability, but even then, since $p$ is odd, $\pi_{D_1}(c) = \chi^p \pi_{D_0}(c) = \chi^{p+1} = 1$, so that $(\cA,\cH,D_1)$ is, in fact, orientable. Hence, we may apply the Reconstruction Theorem to $(\cA,\cH,D_1)$ to obtain $X$.

Now, by the Serre--Swan theorem, there exists a smooth vector bundle $H$ on $M$ such that $\cHi = C^{\infty}(X,H)$, so that the Hermitian structure on $\cHi$ induces a Hermitian structure on the vector bundle $H$. Moreover, by~\cite{GBVF}*{\S 11.2, Proposition 11.5} and our argument above concerning stability of absolute continuity, we have that
\[
	\forall f \in C^{\infty}(X), \quad \ncint f \abs{D}^{-p} = \int f d\nu,
\]
where $d\nu$ is a constant multiple of the volume form on $X$, so that, in particular, $\cH = L^{2}(X,H,d\nu)$. Finally, by the order one condition~\cite{GBVF}*{p.\ 501}, $D$ is indeed an essentially self-adjoint first-order differential operator on $H$.

At last, suppose in addition that $(\cA,\cH,D)$ is of Dirac type. By this assumption and the order one condition, we can define a positive semi-definite $\bR$-valued quadratic form $Q$ on $T^{\ast}M$ by
\[
	Q(df) := -[D,f]^{2}, \quad f \in C^{\infty}(X,\bR).
\]
Thus, in order to construct a Riemannian metric, with respect to which $df \mapsto [D,f]$ defines a self-adjoint Clifford action on $H$ making $D$ into a Dirac-type operator on $H$, it suffices to show that $Q$ is non-degenerate. However, by~\cite{GBVF}*{p. 504}, \latin{mutatis mutandis}, pre-orientability of $(\cA,\cH,D)$ implies that $Q$ is indeed non-degenerate, and therefore yields the required Riemannian metric.
\end{proof}

At last, we can prove our reconstruction theorem:

\begin{proof}[Proof of Theorem~\ref{acst}]
By our discussion in \S~\ref{acstsec}, it suffices to prove the forward direction. Hence, let $(\cA,\cH,D)$ is an abstract almost-commutative spectral triple with base $\cB$ and metric dimension $p$.

First, by Corollary~\ref{weakrecon} applied to $(\cB,\cH,D)$, there exist a compact oriented Riemannian $p$-manifold $X$ and a self-adjoint Clifford module bundle $H$ over $X$ such that $\cB = C^{\infty}(X)$, $\cH = L^{2}(X,H)$, and $D$ is an essentially self-adjoint Dirac-type operator on $H$.

Next, condition (2) for almost-commutative spectral triples, together with Theorem~\ref{SerreSwan}, implies that $\cA = C^{\infty}(X,A)$ for $A$ a real unital involutive weak algebra bundle over $X$.

Finally, condition (3) for almost-commutative spectral triples, together with the fact that $D$ is a Dirac-type operator on the Clifford module bundle $H$, implies that $A$ can be identified as a real unital \Star-algebra sub-bundle of $\End_{\Cl(X)}^{+}(H)$, as required.
\end{proof}

\section{Next steps}

As we have seen, we can obtain a reconstruction theorem for almost-commutative spectral triples, indeed as a consequence of Connes's reconstruction theorem for commutative spectral triples, but only by expanding the definition of almost-com\-mu\-ta\-tive spectral triples. This generalisation does come with some advantages, though. 

On the one hand, our definition naturally accommodates non-trivial finite ``fibrations'' of spectral triples, examples of which have already appeared in the literature~\cite{BvS10}. If product almost-commutative spectral triples offer the simplest non-trivial example of a product of spectral triples (even if we do not yet have a standard category of spectral triples), then perhaps the ``twisted products'' of Section~\ref{cast} might be seen as the simplest non-trivial examples of a non-trivial fibration with fixed fibre in a candidate category of spectral triples, whilst inner fluctuations of the metric of such almost-commutative spectral triples might be seen as giving examples of more slightly general fibrations. Certainly, as we have observed above, given a compact spin manifold $X$, a finite spectral triple $F$, and twisted product $X \times_{(\cP,\nabla_{P})} F$ of $X$ and $F$, there is an obvious morphism $X \times_{(\cP,\nabla_{P})} F \to X$ in the $KK$-theoretic category proposed by Mesland~\cites{Me09,Me09b}. The significance of this example in defining fibrations in Mesland's category, however, remains to be seen.

On the other hand, our definition is manifestly global-analytic, based emphatically on the theory of Dirac-type operators, and is stable under inner fluctuation of the metric. As such, it lends itself readily to a very general study of the perturbative spectral action on almost-commutative spectral triples, using all of the global-analytic computational techniques used thus far on product almost-commutative spectral triples~\cite{vdD11}. It also lends itself to the systematic use of superconnections, whose relevance to the study of Dirac-type operators has been established by the monograph of Berline--Getzler--Vergne~\cite{BGV}; indeed, the use of superconnections in understanding (and even formulating) the noncommutative geometric Standard Model was already proposed by Figueroa--Gracia-Bond\'\i a--Lizzi--V\'arilly in 1998~\cite{FGBLV}. However, a full understanding of the perturbative spectral action on almost-commutative spectral triples will require the systematic study of real structures on such spectral triples.

Finally, the global analytic nature of our definition implies that $KK$-theoretic aspects of almost-commutative spectral triples can be studied in particular detail. Such work has already been undertaken by Zhang for a specific class of  almost-commutative spectral triples he calls \term{projective spectral triples}~\cite{Zha10}.

\appendix
\section{Stability results}

In this appendix, we provide proofs of the stability results needed to prove Corollary~\ref{weakrecon}, and hence Theorem~\ref{acst}. Our account generally follows that of~\cite{ChM09}.

In what follows, let $(\cA,\cH,D)$ be a spectral triple, let $M \in B(\cH)$ be self-adjoint, and let $D_{M} = D+M$. Before continuing, let us prove the following basic stability result:

\begin{lemma}[\cite{ChM09}*{Lemma 2.1, Proposition 2.2}]\label{lem1}
One has that $(\cA,\cH,D_{M})$ is a spectral triple.
\end{lemma}

\begin{proof}
First, by the Kato-Rellich theorem~\cite{RS2}*{Theorem X.12}, $D_{M}$ is self-adjoint on $\dom D_{M}=\dom D$ and essentially self-adjoint on any core of $D$. Next, since $D$ has compact resolvent, for any $\lambda \in \bC \setminus \bR$,
\[
	(D_{M}-\lambda)^{-1} = (D-\lambda)^{-1} - (D_{M}-\lambda)^{-1}M(D-\lambda)^{-1} \in \cK(\cH),
\]
so that $D_{M}$ too has compact resolvent. Finally, for any $a \in \cA$, since $[D,a] \in B(\cH)$ and since $M \in B(\cH)$, 
\[
	[D_{M},a] = [D,a] + [M,a] \in B(\cH).
\]
Thus, $(\cA,\cH,D_{M})$ is indeed a spectral triple.
\end{proof}

\subsection{Metric dimension}

Let us now consider stability of metric dimension.

\begin{lemma}\label{lem2}
If $(\cA,\cH,D)$ has metric dimension $p$, then so does $(\cA,\cH,D_{M})$.
\end{lemma}

\begin{proof}
By~\cite{CP98}*{Lemma B.6}, one has that
\[
	\frac{1}{f(\norm{M})}(D^{2}+1)^{-1} \leq (D_{M}^{2}+1)^{-1} \leq f(\norm{M})(D^{2}+1)^{-1},
\]
where $f(x) := 1 + \tfrac{1}{2}x^{2} + \tfrac{1}{2}x\sqrt{x^{2}+4}$, so that by~\cite{RS4}*{Lemma on p.\ 270}, if $\lambda_{n}(C)$ denotes the $n$-th eigenvalue of the positive compact operator $C \in B(\cH)$, in decreasing order, then
\[
	\frac{1}{f(\norm{M})}\lambda_{k}((D^{2}+1)^{-1}) \leq \lambda_{k}((D_{M}^{2}+1)^{-1}) \leq f(\norm{M})\lambda_{k}((D^{2}+1)^{-1})
\]
for all $n \in \bN$. Since $\lambda_{k}((D^{2}+1)^{-1}) = O(k^{-2/p})$, it therefore follows that 
\[
	\lambda_{k}((D_{M}^{2}+1)^{-1}) = O(k^{-2/p}).
\]
Thus, $(A,\cH,D_{M})$ has metric dimension $p$, as was claimed.
\end{proof}

\subsection{Finiteness}

The following lemma will suffice to establish stability of finiteness in the proof of Corollary~\ref{weakrecon}, and will also be necessary for our discussion below of [strong] regularity and absolute continuity; we shall follow the proof by Iochum--Levy--Vassilevich.

\begin{lemma}[\cite{ILV}*{Lemma 2.3}]\label{lem3}
For $k \in \bN$, let $H^{k} := \dom D^{k}$ with the Sobolev inner product 
\[
	\ip{\xi,\eta}_{k} := \ip{D^{k}\xi,D^{k}\eta} + \ip{\xi,\eta},
\]
and similarly let $H^{k}_{M} := \dom D^{k}_{M}$ with the Sobolev inner product
\[
	\ip{\xi,\eta}_{M,k} := \ip{D_{M}^{k}\xi,D_{M}^{k}\eta} + \ip{\xi,\eta}.
\]
Suppose now that $M$ restricts to an element of $B(H^{k})$ for each $k \in \bN$. Then $H^{k} = H^{k}_{M}$ for all $k \in \bN$ with equivalent norms, and thus, in particular, $\cap_{k} \dom D^{k} = \cap_{k} \dom D^{k}_{M}$.
\end{lemma}

\begin{proof}
Let us first prove equality of vector spaces. We proceed by induction on $k$.  First, by the Kato-Rellich theorem~\cite{RS2}*{Theorem X.12}, $D_{M}$ is self-adjoint on $\dom D_{M}=\dom D$ and essentially self-adjoint on any core of $D$, so that the claim holds for $k=1$ . Now, assume by induction that the claim holds for some $m \in \bN$. Then, by the induction hypothesis and our restriction on $M$,
\begin{align*}
	\dom D_{M}^{m+1} &= \set{\xi \in \dom D_{M}^{m} \mid D_{M}\xi \in \dom D_{M}^{m}}\\
	&= \set{\xi \in \dom D^{m} \mid (D+M)\xi \in \dom D^{m}}\\
	&= \set{\xi \in \dom D^{m} \mid D\xi \in \dom D^{m}}\\
	&= \dom D^{m+1},
\end{align*}
as required.

Let us now prove equivalence of the Sobolev norms. Before continuing, we will find it convenient to replace $\ip{\cdot,\cdot}_{k}$ and $\ip{\cdot,\cdot}_{M,k}$ with $\hp{\cdot,\cdot}_{k}$ and $\hp{\cdot,\cdot}_{M,k}$, respectively, where
\begin{align*}
	\hp{\xi,\eta}_{k} &:= \ip{(D+i)^{k}\xi,(D+i)^{k}\eta} + \ip{\xi,\eta},\\
	\hp{\xi,\eta}_{M,k} &:= \ip{(D_{M}+i)^{k}\xi,(D_{M}+i)^{k}\eta} + \ip{\xi,\eta}.
\end{align*}
Indeed, let us show, for instance, that $\ip{\cdot,\cdot}_{k}$ and $\hp{\cdot,\cdot}_{k}$ define equivalent norms. On the one hand, for $\xi \in H^{k}$,
\[
	\norm{(D+i)^{k}\xi} = \norm{\sum_{m=0}^{k} i^{m} D^{k-m}\xi} \leq \sum_{m=0}^{k} \norm{D^{k-m}\xi},
\]
so that by continuity of the inclusions $H^{k} \hookrightarrow H^{k-m}$ for the $\ip{\cdot,\cdot}_{n}$, there exists some $C>0$, independent of $\xi$, such that
\[
	\norm{(D+i)^{k}\xi}^{2} + \norm{\xi}^{2} \leq C\left(\norm{D^{k}\xi}^{2} + \norm{\xi}^{2}\right).
\]
On the other hand since the $\hp{\cdot,\cdot}_{k}$ is also simply the $k$-th Sobolev inner product for $\sqrt{D^{2}+1}$, that the inclusions $H^{k} \hookrightarrow H^{k-m}$ are also continuous for the $\hp{\cdot,\cdot}_{n}$, and hence, since
\[
	\norm{D^{k}\xi} = \norm{((D+i)-i)^{k}\xi} =  \norm{\sum_{m=0}^{k} (-i)^{m} (D+i)^{k-m}\xi} \leq \sum_{m=0}^{k} \norm{(D+i)^{k-m}\xi},
\]
there exists some $C^{\prime}$, independent of $\xi$, such that
\[
	\norm{D^{k}\xi}^{2} + \norm{\xi}^{2} \leq C^{\prime}\left(\norm{(D+i)^{k}\xi}^{2} + \norm{\xi}^{2}\right).
\]
Thus, $\ip{\cdot,\cdot}_{k}$ and $\hp{\cdot,\cdot}_{k}$ do indeed define equivalent norms.

Now, fix $k \in \bN$, and consider the linear map $B = (D_{M}-i)^{k}(D-i)^{-k}$ on $\cH$; we claim that $B$ is, in fact, bounded on $\cH$. First, one has that on $\dom D^{k} = \dom D_{M}^{k}$,
\[
	(D_{M}-i)^{k} = ((D-i)+M)^{k} = (D-i)^{k} + \sum_{m=1}^{k} T_{m},
\]
where for each $m$, $T_{m}$ is a product of $k$ operators, each of which is either $(D-i)$ or $M$. By our assumption on $M$, then, each $T_{m}$ therefore defines a continuous map $H_{k} \to H_{1}$, so that since $(D-i)^{k} : \cH_{k} \to \cH$ and $(D-i)^{-k} : \cH \to \cH_{k}$ are continuous,
\[
	B = (D_{M}-i)^{k}(D-i)^{-k} = \Id_{\cH} + \sum_{m=1}^{k} T_{m}(D-i)^{k}
\]
defines a bounded operator on $\cH$. Since $B$ is bijective, it therefore follows by the bounded inverse theorem that $B$ has a bounded inverse. Thus, for $\xi \in \dom D^{k} = \dom D_{M}^{k}$, since $(D_{M}-i)^{k} = B(D-i)^{k}$ and $(D-i)^{k} = \inv{B}(D_{M}-i)^{k}$,
\[
	\hp{\xi,\xi}_{M,k} \leq \max\set{1,\norm{B}^{2}}\hp{\xi,\xi}_{k}, \quad \hp{\xi,\xi}_{k} \leq \max\set{1,\norm{\inv{B}}^{2}}\hp{\xi,\xi}_{M,k},
\]
which implies, by our earlier observation, that $\norm{\cdot}_{k}$ and $\norm{\cdot}_{M,k}$ are equivalent, as required.
\end{proof}

\subsection{[Strong] Regularity}

We shall now use Higson's characterisation of [strong] regularity, first to prove a corollary thereof that will allow us actually to apply Lemma~\ref{lem3} in context, and then to prove a result on stability of [strong] regularity, due to Chakraborty and Mathai~\cite{ChM09}.

In order to state and use Higson's characterisation, we shall need the following definition due to Higson~\cite{Hig06}:

\begin{definition}\label{algdif}
Suppose that each $a \in A$ maps $\cHi = \cap_{k} \dom D^{k}$ to itself. The \emph{[extended] algebra of differential operators} associated to $(\cA,\cH,D)$ is the smallest algebra $\cD$ of linear operators on $\cHi$ which contains $\cA$ and $[D,\cA]$ [and $\End_{\cA}(\cHi)$], and which is closed under the operator $T \mapsto [D^{2},T]$. The algebra $\cD$ is defined and filtered inductively as follows:
\begin{itemize}
	\item[(a)] $\cD_{0}$ is the algebra generated by $\cA$ and $[D,\cA]$ [and $\End_{\cA}(\cHi)$].
	\item[(b)] $\cD_{1} = [D^{2},\cD_{0}] + \cD_{0}[D^{2},\cD_{0}]$.
	\item[(c)] $\cD_{k} = \sum_{j=1}^{k-1}\cD_{j}\cD_{k-j} + [D^{2},\cD_{k-1}] + \cD_{0}[D^{2},\cD_{k-1}]$.
\end{itemize}
Moreover, we shall call $(\cD,D)$ a \term{differential pair} if for every $X \in \cD$ of order $\leq k$, there exists some $\epsilon > 0$ such that for all $\xi \in \cHi$, $\norm{D^{k}\xi} + \norm{\xi} \geq \epsilon \norm{X\xi}$.
\end{definition} 

In the case of a spectral triple of the form $(C^{\infty}(X),L^{2}(X,H),D)$, where $X$ is a compact orientable Riemannian manifold, $H$ is a self-adjoint Clifford module bundle over $X$, and $D$ is a symmetric Dirac-type operator on $H$, then $\cD$, in either case, is indeed an algebra of differential operators on $H$; that $(\cD,D)$ is a differential pair follows from the G\r arding estimates of elliptic regularity theory~\cite{Roe}*{Chapter 5}.

Now we can give Higson's characterisation of [strong] regularity:

\begin{theorem}[Higson~\cite{Hig06}*{Theorem 4.26}]\label{Higson}
Suppose that each $a \in A$ maps $\cHi = \cap_{k} \dom D^{k}$ to itself. Let $\cD$ be the [extended] algebra of differential operators associated to $(\cA,\cH,D)$. Then $(\cA,\cH,D)$ is [strongly] regular if and only if $(\cD,D)$ is a differential pair.
\end{theorem}

The above theorem was originally stated by Higson as a characterisation of regularity, but his proof immediately applies to yield the analogous characterisation of strong regularity.

Let us now turn to our corollary of Theorem~\ref{Higson}, which will allow us to use Lemma~\ref{lem3}:

\begin{corollary}\label{lem4a}
Suppose that $(\cA,\cH,D)$ is [strongly] regular with [extended] algebra of differential operators $\cD$, and that $M \in \cD_{0}$. Let $\cH^{k} := \dom D^{k}$ with the usual Sobolev inner product. Then for each $k \in \bN$, $M$ restricts to a bounded operator on $\cH^{k}$.
\end{corollary}

Since $(\cD,D)$ is a differential pair by Theorem~\ref{Higson}, this corollary is an immediate consequence of the following lemma of Higson's, proved by him in a more general context:

\begin{lemma}[\cite{Hig06}*{Lemma 4.7}]\label{lem4aa}
Suppose that each $a \in A$ maps $\cHi = \cap_{k} \dom D^{k}$ to itself. Let $\cD$ be the [extended] algebra of differential operators associated to the spectral triple $(\cA,\cH,D)$. For each $k \in \bN \cup \set{0}$, let $\cH^{k} := \dom D^{k}$ with the Sobolev inner product. If $(\cD,D)$ is a differential pair, then for any $X \in \cD_{k}$, $X$ extends to a bounded operator $\cH^{k+m} \to \cH^{m}$ for all $m \in \bN \cup \set{0}$.
\end{lemma}

Finally, with Higson's characterisation of [strong] regularity and Lemma~\ref{lem4a} in place, we can finally state and prove our stability result for [strong] regularity:

\begin{lemma}[\cite{ChM09}*{Proposition 4.2}]\label{lem4}
Suppose that $(\cA,\cH,D)$ is [strongly] regular with [extended] algebra of differential operators $\cD$, and that $M \in \cD_{0}$. Suppose, moreover, that $[D_{M}^{2}-D^{2},T] \in \cD_{k+1}$ for $T \in \cD_{k}$. Then $(\cA,\cH,D_{M})$ is also [strongly] regular.
\end{lemma}

\begin{proof}[Proof of Lemma~\ref{lem4}]
First, by regularity of $(\cA,\cH,D)$ and Lemma~\ref{lem3}, $(\cA,\cH,D_{M})$ is such that each $a \in \cA$ maps $\cap_{k} \dom D_{M}^{k} = \cap_{k} \dom D^{k} = \cHi$ to itself.

Now, let $\cD_{M}$ be the [extended] algebra of differential operators for $(\cA,\cH,D_{M})$. Then, by our hypothesis on $D_{M}^{2}-D^{2}$,  there is a filtered inclusion $\cD_M \subset \cD$ of [extended] algebras of differential operators, so that by [strong] regularity of $(\cA,\cH,D)$ and Theorem~\ref{Higson}, $\cD_{M}$ satisfies the basic estimate for $D$. Fix $X \in \cD_{M}$ of order $\leq k$, so that there exists some $\epsilon > 0$ such that for all $\xi \in \cHi := \cap_{m} \dom D^{m} = \cap_{m} \dom D_{M}^{m}$,
\[
	\norm{D^{k}\xi} + \norm{\xi} \geq \epsilon\norm{X\xi}.
\]
However, since $\dom D^{k} = \dom D_{M}^{k}$ with equivalent Sobolev norms by Lemmas~\ref{lem4a} and~\ref{lem3}, it follows that that $D^{k}$ is bounded as an operator from $\dom D_{M}^{k}$ endowed with the Sobolev $k$-norm for $D_{M}$, to $\cH$, implying that $\norm{D^{k}\xi} \leq \alpha \norm{D_{M}^{k}\xi} + \beta \norm{\xi}$ for some $\alpha$, $\beta > 0$ independent of $\xi$, and hence that
\[
	\norm{D^{k}_{M}\xi} + \norm{\xi} \geq \epsilon^{\prime}\norm{X\xi}
\]
for some $\epsilon^{\prime} > 0$ independent of $\xi$. Thus, $(\cD_{M},D_{M})$ is a differential pair, so that by Theorem~\ref{Higson}, $(\cA,\cH,D_{M})$ is indeed [strongly] regular.
\end{proof}

\begin{remark}
If $M$ is an inner fluctuation of the metric, that is, if 
\[
	M = \sum_{i=1}^{n} a_{i}[D,b_{i}]
\]
for some $a_{i}$, $b_{i} \in \cA$, then the condition that $[D_{M}^{2}-D^{2},T] \in \cD_{k+1}$ for $T \in \cD_{k}$ is automatically satisfied.
\end{remark}

\subsection{Absolute continuity}

Finally, we consider stability of absolute continuity.

\begin{lemma}\label{lem5}
If $(\cA,\cH,D)$ is strongly regular and of metric dimension $p$, and if $M \subset \End_{\cA}(\cHi)$ for $\cHi = \cap_{k} \dom D^{k}$, then for all $T \in B(\cH)$, $\ncint T \abs{D}^{-p} = \ncint T \abs{D_{M}}^{-p}$.
\end{lemma}

Before continuing to the proof, we will need the following lemma, which will allow us to use $(D^{2}+1)^{-p/2}$ and $(D^{2}_{M}+1)^{-p/2}$ instead of $\abs{D}^{-p}$ and $\abs{D_{M}}^{-p}$, respectively, making calculations easier in the case that $D$ or $D_{M}$ is not invertible.

\begin{lemma}\label{lem5a}
Let $T$ is a positive operator on $\cH$ with compact resolvent, and let $\mu_{k}$ denote the $k$-th eigenvalue of $T$ in increasing order, counted with multiplicity. Suppose that $\mu_{k} = O(k^{2/p})$ as $k \to +\infty$ for some $p \in \bN$. Then $T^{-p/2} - (T+1)^{-p/2}$ is trace-class.
\end{lemma}

\begin{proof}
First, observe that $T^{-p/2} - (T+1)^{-p/2}$ has eigenvalues
\[
	\lambda_{k} =
	\begin{cases}
		-(\mu_{k}+1)^{-p/2} &\text{if $k \leq \dim \ker T$,}\\
		\mu_{k}^{-p/2} - (\mu_{k}+1)^{-p/2} &\text{if $k > \dim \ker T$,}\\
	\end{cases}
\]
counted with multiplicity. Then for $k \gg \dim \ker T$, so that $\lambda_{k} \geq 0$ and $\mu_{k} > 1$,
\begin{align*}
	\lambda_{k} &= \frac{1}{\mu_{k}^{p/2}} - \frac{1}{(\mu_{k}+1)^{p/2}}\\ 
	&= \frac{\sum_{m=0}^{p-1} \binom{p}{m} \mu_{k}^{m}}{\mu_{k}^{p/2}(\mu_{k}+1)^{p/2}(\mu_{k}^{p/2}+(\mu_{k}+1)^{p/2})}\\
	&\leq \frac{1}{2} \left(\sum_{m=0}^{p-1} \binom{p}{m}\right) \mu_{k}^{-(p+2)/2}\\
	&= O(k^{-(1+2/p)})
\end{align*}
as $k \to +\infty$. Hence, $T^{-p/2}-(T+1)^{-p/2}$ is trace-class, as was claimed.
\end{proof}

In what follows, let $\cL^{1+}(\cH)$ denote the Dixmier trace-class ideal of $B(\cH)$, and let $\cL^{k}(\cH)$ denote the $k$-th Schatten ideal of $B(\cH)$; for further background, see~\cite{GBVF}*{Chapter 7}.

\begin{proof}[Proof of Lemma~\ref{lem5}]
Let
\[
	n = \begin{cases} \tfrac{p}{2} &\text{if $p$ is even,}\\ \tfrac{p+1}{2} &\text{if $p$ is odd}, \end{cases} \quad \alpha = \frac{2p}{n} = \begin{cases} 1 &\text{if $p$ is even,}\\ \tfrac{p}{p+1} &\text{if $p$ is odd.} \end{cases}
\]
First, we have that
\begin{align*}
	&(D_{M}^{2}+1)^{-n} - (D^{2}+1)^{-n}\\
	=\ &(D^{2}+1)^{-n}\left( (D^{2}+1)^{n} - (D_{M}^{2}+1)^{n} \right) (D_{M}^{2}+1)^{-n}\\
	=\ &(D^{2}+1)^{-n}\left(\sum_{i=0}^{n} \sum_{j=0}^{2i-1} {n \choose i} D^{j} M D_{M}^{2i-1-j}  \right) (D_{M}^{2}+1)^{-n}\\
	=\ &\sum_{i=0}^{n}\sum_{j=0}^{2i-1} {n \choose i} (D^{2}+1)^{-n+\frac{j}{2}}\left[D(D^{2}+1)^{-\frac{1}{2}}\right]^{j} M \left[D_{M}(D_{M}^{2}+1)^{-\frac{1}{2}}\right]^{2i-1-j} (D^{2}_{M}+1)^{-n+i-\frac{j+1}{2}},
\end{align*}
which, by Corollary~\ref{lem4a} and Lemma~\ref{lem3} can be checked on the common core $\cHi$ of $D$ and $D_{M}$.

Now, consider the term
\[
	(D^{2}+1)^{-n+\frac{j}{2}}\left[D(D^{2}+1)^{-\frac{1}{2}}\right]^{j} M \left[D_{M}(D_{M}^{2}+1)^{-\frac{1}{2}}\right]^{2i-1-j} (D^{2}_{M}+1)^{-n+i-\frac{j+1}{2}},
\]
where $0 \leq i \leq n$ and $0 \leq j \leq 2i-1$. Since $(\cA,\cH,D)$ is of metric dimension $p$, so too is $(\cA,\cH,D_{M})$ by Lemma~\ref{lem2}, so that
\[
	(D^{2}+1)^{-p/2}, \; (D^{2}_{M}+1)^{-p/2} \in \cL^{1+}(\cH),
\]
and hence, for all $\epsilon > 0$,
\[
	(D^{2}+1)^{-1}, \; (D^{2}_{M}+1)^{-1} \in \cL^{\alpha n+\epsilon}(\cH).
\]
Setting $\epsilon = \alpha(n-i+\tfrac{1}{2})$, we therefore find that
\[
	(D^{2}+1)^{-n+\frac{j}{2}} \in \cL^{r}(\cH),
\]
\[
	\left[D(D^{2}+1)^{-\frac{1}{2}}\right]^{j} M \left[D_{M}(D_{M}^{2}+1)^{-\frac{1}{2}}\right]^{2i-1-j} (D^{2}_{M}+1)^{-n+i-\frac{j+1}{2}} \in \cL^{s}(\cH)
\]
for
\[
	r = \frac{\alpha n+\epsilon}{n-\frac{j}{2}}, \quad s = \frac{\alpha n+\epsilon}{n-i+\frac{j+1}{2}},
\]
which satisfy $r^{-1}+s^{-1}=\alpha^{-1}$. Hence, by H{\"o}lder's inequality for Schatten norms,
\[
	(D^{2}+1)^{-n+\frac{j}{2}}\left[D(D^{2}+1)^{-\frac{1}{2}}\right]^{j} M \left[D_{M}(D_{M}^{2}+1)^{-\frac{1}{2}}\right]^{2i-1-j} (D^{2}_{M}+1)^{-n+i-\frac{j+1}{2}} \in \cL^{\alpha}(\cH);
\]
since this is true for all $i$ and $j$, it therefore follows that $(D_{M}^{2}+1)^{-n} - (D^{2}+1)^{-n} \in \cL^\alpha(\cH)$. If $p$ is even, then $\alpha = 1$ and $(D_{M}^{2}+1)^{-n} - (D^{2}+1)^{-n}$ is already trace-class; if $p$ is odd, then since
\[
	(D_M^2+1)^{-p/2\alpha} - (D^2+1)^{-p/2\alpha} = (D_{M}^{2}+1)^{-n} - (D^{2}+1)^{-n} \in \cL^\alpha(\cH)
\]
for $0 < \alpha < 1$, we can apply the BKS inequality~\cite{BKS}*{Thm.~1} to \[\abs{(D_M^2+1)^{-p/2\alpha} - (D^2+1)^{-p/2\alpha}}^\alpha \in \cL^{1}(\cH)\] to conclude that $(D^2_M+1)^{-p/2} - (D^2+1)^{-p/2}$ is indeed trace-class. Either way, from these calculations and Lemma~\ref{lem5a}, it therefore follows that $\abs{D_{M}}^{-p} - \abs{D}^{-p}$ is trace-class, yielding the desired result.
\end{proof}


\begin{bibdiv}
\begin{biblist}
\bib{Bar07}{article}{
   author={Barrett, John W.},
   title={Lorentzian version of the noncommutative geometry of the Standard
   Model of particle physics},
   journal={J. Math. Phys.},
   volume={48},
   date={2007},
   number={1},
}
\bib{BGV}{book}{
   author={Berline, Nicole},
   author={Getzler, Ezra},
   author={Vergne, Mich{\`e}le},
   title={Heat kernels and Dirac operators},
   series={Grundlehren Text Editions},
   note={Corrected reprint of the 1992 original},
   publisher={Springer-Verlag},
   place={Berlin},
   date={2004},
}
\bib{BKS}{article}{
   author={Birman, M. {\v{S}}.},
   author={Koplienko, L. S.},
   author={Solomjak, M. Z.},
   title={Estimates of the spectrum of a difference of fractional powers of
   selfadjoint operators},
   journal={Izv. Vys\v s. U\v cebn. Zaved. Matematika},
   date={1975},
   number={3(154)},
   pages={3--10},
}
\bib{BvS10}{article}{
	author={Boeijink, Jord},
	author={van Suijlekom, Walter D.},
	title={The noncommutative geometry of Yang-Mills fields},
	journal={J.\ Geom.\ Phys.},
	volume={61},
	date={2011},
	number={6},
	pages={1122-1134},
}
\bib{Ca10}{article}{
	author={\'Ca\'ci\'c, Branimir},
	title={Moduli spaces of Dirac operators for finite spectral triples},
	conference={
		title={Quantum groups and noncommutative spaces},
	},
	book={
		editor={Marcolli, Matilde},
		editor={Parashar, Deepak},
		series={Aspects Math.},
		volume={41},
		publisher={Vieweg + Teubner, Wiesbaden},
	},
	date={2011},
	pages={9--68},
}
\bib{CP98}{article}{
   author={Carey, Alan},
   author={Phillips, John},
   title={Unbounded Fredholm modules and spectral flow},
   journal={Canad.\ J.\ Math.},
   volume={50},
   date={1998},
   number={4},
   pages={673--718},
}
\bib{ChM09}{article}{
   author={Chakraborty, Partha Sarathi},
   author={Mathai, Varghese},
   title={The geometry of determinant line bundles in noncommutative
   geometry},
   journal={J. Noncommut. Geom.},
   volume={3},
   date={2009},
   number={4},
   pages={559--578},
}
\bib{CC97}{article}{
   author={Chamseddine, Ali H.},
   author={Connes, Alain},
   title={The spectral action principle},
   journal={Comm. Math. Phys.},
   volume={186},
   date={1997},
   number={3},
   pages={731--750},
}
\bib{CC06}{article}{
   author={Chamseddine, Ali H.},
   author={Connes, Alain},
   title={Inner fluctuations of the spectral action},
   journal={J. Geom. Phys.},
   volume={57},
   date={2006},
   number={1},
   pages={1--21},
}
\bib{CCM07}{article}{
   author={Chamseddine, Ali H.},
   author={Connes, Alain},
   author={Marcolli, Matilde},
   title={Gravity and the standard model with neutrino mixing},
   journal={Adv. Theor. Math. Phys.},
   volume={11},
   date={2007},
   number={6},
   pages={991--1089},
}
\bib{Con88}{article}{
   author={Connes, A.},
   title={The action functional in noncommutative geometry},
   journal={Comm. Math. Phys.},
   volume={117},
   date={1988},
   number={4},
   pages={673--683},
}
\bib{Con95}{article}{
   author={Connes, Alain},
   title={Geometry from the spectral point of view},
   journal={Lett. Math. Phys.},
   volume={34},
   date={1995},
   number={3},
   pages={203--238},
}
\bib{Con95b}{article}{
   author={Connes, Alain},
   title={Noncommutative geometry and reality},
   journal={J. Math. Phys.},
   volume={36},
   date={1995},
   number={11},
   pages={6194--6231},
}
\bib{Con96}{article}{
   author={Connes, Alain},
   title={Gravity coupled with matter and the foundation of non-commutative
   geometry},
   journal={Comm. Math. Phys.},
   volume={182},
   date={1996},
   number={1},
   pages={155--176},
}
\bib{Con06}{article}{
   author={Connes, Alain},
   title={Noncommutative geometry and the standard model with neutrino
   mixing},
   journal={J. High Energy Phys.},
   date={2006},
   number={11},
}
\bib{Con08}{article}{
	author={Connes, Alain},
	title={On the spectral characterization of manifolds},
	date={2008},
	eprint={arXiv:0810.2088v1 [math.OA]}
}
\bib{Con08b}{article}{
	author={Connes, Alain},
	title={The spectral characterization of manifolds},
	conference={
		title={Distinguished Lecture, Thematic Program on Operator Algebras},
		address={Fields Institute, Toronto, ON, Canada},
		date={2008-05-28},
	},		
}
\bib{CoMa}{book}{
   author={Connes, Alain},
   author={Marcolli, Matilde},
   title={Noncommutative geometry, quantum fields and motives},
   series={American Mathematical Society Colloquium Publications},
   volume={55},
   publisher={American Mathematical Society},
   place={Providence, RI},
   date={2008},
}
\bib{CMo95}{article}{
   author={Connes, A.},
   author={Moscovici, H.},
   title={The local index formula in noncommutative geometry},
   journal={Geom. Funct. Anal.},
   volume={5},
   date={1995},
   number={2},
   pages={174--243},
}
\bib{vdD11}{thesis}{
	author={van den Dungen, K.},
	title={The structure of gauge theories in almost-commutative geometries},
	type={Master's thesis},
	organization={Radbout University Nijmegen},
	date={2011}
}
\bib{FGBLV}{article}{
   author={Figueroa, H.},
   author={Gracia-Bond{\'{\i}}a, J. M.},
   author={Lizzi, F.},
   author={V{\'a}rilly, J. C.},
   title={A nonperturbative form of the spectral action principle in
   noncommutative geometry},
   journal={J. Geom. Phys.},
   volume={26},
   date={1998},
   number={3-4},
   pages={329--339},
}
\bib{Gil}{book}{
   author={Gilkey, Peter B.},
   title={Invariance theory, the heat equation, and the Atiyah-Singer index
   theorem},
   series={Studies in Advanced Mathematics},
   edition={2},
   publisher={CRC Press},
   place={Boca Raton, FL},
   date={1995},
}
\bib{GBVF}{book}{
   author={Gracia-Bond{\'{\i}}a, Jos{\'e} M.},
   author={V{\'a}rilly, Joseph C.},
   author={Figueroa, H{\'e}ctor},
   title={Elements of noncommutative geometry},
   series={Birkh\"auser Advanced Texts: Basler Lehrb\"ucher},
   publisher={Birkh\"auser Boston Inc.},
   place={Boston, MA},
   date={2001},
}
\bib{Hig06}{article}{
   author={Higson, Nigel},
   title={The residue index theorem of Connes and Moscovici},
   conference={
      title={Surveys in noncommutative geometry},
   },
   book={
      editor={Higson, Nigel},
      editor={Roe, John},
      series={Clay Math. Proc.},
      volume={6},
      publisher={Amer. Math. Soc.},
      place={Providence, RI},
   },
   date={2006},
   pages={71--126},
}
\bib{ILV}{article}{
	title={Spectral action beyond the weak-field approximation},
	author={Iochum, B.},
	author={Levy, C.},
	author={Vassilevich, D.},
	date={2011},
	eprint={arXiv:1108.3749v1 [hep-th]}
}
\bib{ACG1}{article}{
	title={On a classification of irreducible almost commutative geometries},
	author={Iochum, Bruno},
	author={Sch{\"u}cker, Thomas},
	author={Stephan, Christoph},
	journal={J. Math. Phys.},
	volume={45},
	date={2004},
	number={12},
	pages={5003--5041},
}
\bib{ACG2}{article}{
	title={On a classification of irreducible almost commutative geometries, a second helping},
	author={Jureit, Jan-H.},
	author={Stephan, Christoph A.},
	journal={J. Math. Phys.},
	volume={46},
	date={2005},
	number={4},
}
\bib{ACG3}{article}{
	title={On a classification of irreducible almost commutative geometries III},
	author={Jureit, Jan-Hendrik},
	author={Sch{\"u}cker, Thomas},
	author={Stephan, Christoph},
	journal={J. Math. Phys.},
	volume={46},
	date={2005},
	number={7},
}
\bib{ACG4}{article}{
	title={On a classification of irreducible almost commutative geometries IV},
	author={Jureit, Jan-Hendrik},
	author={Stephan, Christoph A.},
	journal={J. Math. Phys.},
	volume={49},
	date={2008},
	number={3}
}
\bib{ACG5}{article}{
	title={On a classification of irreducible almost commutative geometries, V},
	author = {Jureit, Jan-Hendrik},
	author={Stephan, Christoph A.},
	journal={J. Math. Phys.},
	volume={50},
	date={2009},
	number={7}
}
\bib{Kra98}{article}{
   author={Krajewski, Thomas},
   title={Classification of finite spectral triples},
   journal={J. Geom. Phys.},
   volume={28},
   date={1998},
   number={1-2},
   pages={1--30},
}
\bib{Me09}{book}{
   author={Mesland, Bram},
   title={Bivariant K-theory of groupoids and the noncommutative geometry of limit sets},
   series={Bonner Mathematische Schriften},
   volume={394},
   publisher={Universit\"at Bonn, Mathematisches Institut},
   place={Bonn},
   date={2009},
}
\bib{LM}{book}{
   author={Lawson, H. Blaine, Jr.},
   author={Michelsohn, Marie-Louise},
   title={Spin geometry},
   publisher={Princeton University Press},
   place={Princeton, NJ},
   date={1989},
}
\bib{Me09b}{article}{
	author={Mesland, Bram},
	title={Unbounded bivariant $K$-theory and correspondences in noncommutative geometry},
	date={2009},
	eprint={arXiv:0904.4383v2 [math.KT]},
}
\bib{PS98}{article}{
   author={Paschke, Mario},
   author={Sitarz, Andrzej},
   title={Discrete spectral triples and their symmetries},
   journal={J. Math. Phys.},
   volume={39},
   date={1998},
   number={11},
   pages={6191--6205},
}
\bib{PR}{article}{
   author={Prokhorenkov, Igor},
   author={Richardson, Ken},
   title={Perturbations of Dirac operators},
   journal={J. Geom. Phys.},
   volume={57},
   date={2006},
   number={1},
   pages={297--321},
}
\bib{RS2}{book}{
   author={Reed, Michael},
   author={Simon, Barry},
   title={Methods of modern mathematical physics. II. Fourier analysis,
   self-adjointness},
   publisher={Academic Press [Harcourt Brace Jovanovich Publishers]},
   place={New York},
   date={1975},
}
\bib{RS4}{book}{
   author={Reed, Michael},
   author={Simon, Barry},
   title={Methods of modern mathematical physics. IV. Analysis of operators},
   publisher={Academic Press [Harcourt Brace Jovanovich Publishers]},
   place={New York},
   date={1978},
}
\bib{RV06}{article}{
	author={Rennie, Adam},
	author={V\'arilly, Joseph C.},
	title={Reconstruction of manifolds in noncommutative geometry},
	date={2006},
	eprint={arXiv:math/0610418v4 [math.OA]},
}
\bib{Roe}{book}{
   author={Roe, John},
   title={Elliptic operators, topology and asymptotic methods},
   series={Pitman Research Notes in Mathematics Series},
   volume={395},
   edition={2},
   publisher={Longman},
   place={Harlow},
   date={1998},
}
\bib{RV1}{article}{
	author={Roepstorff, G.},
	author={Vehns, Ch.},
	title={An introduction to Clifford supermodules},
	eprint={arXiv:math-ph/9908029v2},
	date={1999}
}
\bib{RV2}{article}{
	author={Roepstorff, G.},
	author={Vehns, Ch.},
	title={Generalized Dirac operators and superconnections},
	eprint={arXiv:math-ph/9911006v1},
	date={1999}
}
\bib{Zha10}{thesis}{
	author={Zhang, Dapeng},
	title={Projective Dirac operators, twisted $K$-theory and local index formula},
	type={Ph.D. dissertation},
	organization={California Institute of Technology},
	date={2011},
}
\end{biblist}
\end{bibdiv}
\end{document}